\documentclass[sigconf]{acmart}
\renewcommand\footnotetextcopyrightpermission[1]{}  
\setcopyright{none}

\usepackage{multirow}
\usepackage{needspace}
\newtheorem{proposition}{Proposition}

\usepackage{cleveref}
\crefname{equation}{}{}

\AtBeginDocument{%
  }

\setcopyright{none}
\copyrightyear{}
\acmYear{}
\acmDOI{}
\acmConference[Manuscript]{Submitted Paper}{August 2025}{Under Review}
\acmISBN{}

\usepackage{makecell}

\begin{document}

\title{Graph Learning for Foreign Exchange Rate Prediction and Statistical Arbitrage}


\author{Yoonsik Hong}
\affiliation{%
  \institution{Northwestern University}
  \city{Evanston}
  \state{Illinois}
  \country{USA}
}
\email{YoonsikHong2028@u.northwestern.edu}

\author{Diego Klabjan}
\affiliation{%
  \institution{Northwestern University}
  \city{Evanston}
  \state{Illinois}
  \country{USA}
}
\email{d-klabjan@northwestern.edu}

\renewcommand{\shortauthors}{Hong and Klabjan}

\begin{abstract}
We propose a two-step graph learning approach for foreign exchange statistical arbitrages (FXSAs), addressing two key gaps in prior studies: the absence of graph-learning methods for foreign exchange rate prediction (FXRP) that leverage multi-currency and currency–interest rate relationships, and the disregard of the time lag between price observation and trade execution. 
In the first step, to capture complex multi-currency and currency–interest rate relationships, we formulate FXRP as an edge-level regression problem on a discrete-time spatiotemporal graph. 
This graph consists of currencies as nodes and exchanges as edges, with interest rates and foreign exchange rates serving as node and edge features, respectively.
We then introduce a graph-learning method that leverages the spatiotemporal graph to address the FXRP problem. 
In the second step, we present a stochastic optimization problem to exploit FXSAs while accounting for the observation-execution time lag. 
To address this problem, we propose a graph-learning method that enforces constraints through projection and ReLU, maximizes risk-adjusted return by leveraging a graph with exchanges as nodes and influence relationships as edges, and utilizes the predictions from the FXRP method for the constraint parameters and node features.   
Moreover, we prove that our FXSA method satisfies empirical arbitrage constraints. 
The experimental results demonstrate that our FXRP method yields statistically significant improvements in mean squared error, and that the FXSA method achieves a 61.89\% higher information ratio and a 45.51\% higher Sortino ratio than a benchmark. 
Our approach provides a novel perspective on FXRP and FXSA within the context of graph learning.

\end{abstract}




\maketitle
\section{Introduction}
\label{sec:intro}
We address the problem of predicting foreign exchange (FX) rates and exploiting FX statistical arbitrage (StatArb) opportunities.  
Consider a currency triplet: A, B, and C.  
We first exchange one unit of currency A for B, then convert the resulting amount of B to C, and finally exchange the resulting amount of C back into A.  
In theory, the final amount of A should be one~\cite{cui2020detecting}, implying inherent interrelationships among FX rates.   
However, in practice, it sometimes exceeds one~\cite{aiba2002triangular}, allowing potential profits.  
In this paper, we predict FX rates by leveraging the interrelationships among currencies and exploit FX StatArbs (FXSAs) involving three or more currencies, which arise in practice.

Although the currencies in the triplet are theoretically interrelated—both among themselves and with interest rates (IRs)—no prior work has studied a graph learning (GL) method to exploit these relationships in FX rate prediction (FXRP).  
As multiple currencies exist globally, combinatorially many triplets can be formed, each inducing its own theoretical interrelations.  
It is difficult to represent such structured relationships in a tabular format and to employ canonical neural architectures, such as multi-layer perceptrons (MLPs). 
Moreover, the interest rate parity (IRP) implies there is an interaction between the IRs of two countries and their FX rate.  
To effectively capture and leverage both triplet-based relationships and the IRP, we leverage GL for FXRP.

While the spot price and the actual trading price can differ due to the time lag between price observation and trade execution, such discrepancies have not been addressed in the existing literature (e.g., \cite{zhang2025efficient,soon2007currency,tatsumura2020currency}) on FX arbitrage. 
This omission renders prior methods potentially inadmissible in practice and subject to look-ahead bias.
Moreover, the price discrepancies introduced by the observation-execution time lag pose risks to investors’ trading, which have likewise been overlooked. 
To account for these discrepancies and the associated risks, we formulate a stochastic optimization problem for FXSA and propose a GL-based method to solve it.

To address multi-currency relationships, currency–IR interactions, and price discrepancies arising from execution time lags, we develop a two-step GL approach. 
In each step, we formulate a problem and set forth a corresponding GL method.

We formulate FXRP as an edge-level regression problem on a discrete-time spatiotemporal graph. 
We define the graph with currencies and exchange relations as nodes and edges, respectively, where IRs and FX rates serve as node and edge features, respectively. 
We train a simple graph neural network (GNN) on this graph.  
The GNN also uses currency-value features we derive from FX rates under the maximum likelihood estimation (MLE) principle.

To exploit FXSAs, we formulate a stochastic optimization problem that accounts for the time lag between price observation and trade execution. 
We then propose a GL method that generates FX trading quantities to address it. 
The GL method solves the problem by enforcing constraints via projection and ReLU, and maximizing risk-adjusted return, while utilizing the predictions from the FXRP method for the constraint parameters and node features. 
We prove our FXSA method guarantees the satisfaction of the constraints. 
We construct a graph with exchanges as nodes and influence relationships—derived from the projection matrix—as edges, and train a GNN on the graph.

We demonstrate the effectiveness of our method on ten of the most traded currencies. 
Our FXSA method outperforms a benchmark across key risk-adjusted return metrics, exhibiting a steady upward trend in cumulative profits while reducing risk. 
Our FXRP results show that incorporating the FX graph with IRs leads to statistically significant improvements over non-graph-based methods.

The prominent contributions of this paper are as follows. 
\begin{itemize}
    \item To the best of our knowledge, we are the first to propose a spatiotemporal GL problem for FXRP using FX rates and government bond IRs, and present a dedicated GL method.
    \item To the best of our knowledge, we are the first to formulate FXSA as a stochastic optimization problem, propose a corresponding GL method that generates trading strategies to exploit arbitrages among three or more currencies, and prove that our FXSA method satisfies the empirical arbitrage constraints, thereby reducing risk. 
    \item We demonstrate the effectiveness of our approach: our FXSA method yields a 61.89\% higher information ratio and a 45.51\% higher Sortino ratio than a benchmark, and the FXRP method shows statistically significant improvement over non-graph models.
\end{itemize}

\section{Related Work}

Existing learning-based studies overlook the interrelationships among FX rates. 
Prior econometric studies on FXRP commonly use models such as 
VAR \cite{carriero2009forecasting} and GARCH \cite{pahlavani2015comparison}. 
However, these methods require underlying statistical assumptions and often fail to capture nonlinear dynamics \cite{ajumi2017exchange}. 
To address these limitations, learning-based methods continue to gain traction. 
Deep learning models (e.g., Transformer \cite{ullah2025forecasting}, LSTM \cite{islam2021foreign}, CNN \cite{ullah2025forecasting,galeshchuk2017deep}) and machine learning algorithms (e.g., gradient boosting \cite{zhao2020wavelet}, SVM \cite{de2018combining,rupasinghe2019forecasting}) are being actively explored. 
Researchers apply such methods to diverse currency pairs (e.g., \texttt{GBPJPY} \cite{rupasinghe2019forecasting}, \texttt{USDCNY} \cite{ullah2025forecasting}, \texttt{USDCAD} \cite{islam2021foreign}), with major pairs like \texttt{USDEUR}~\cite{ahmed2020flf,islam2021foreign,rupasinghe2019forecasting,pradeepkumar2014forex}, 
\texttt{USDJPY}~\cite{pradeepkumar2014forex,rupasinghe2019forecasting}, \texttt{USDGBP}~\cite{islam2021foreign,rupasinghe2019forecasting,pradeepkumar2014forex}, 
\texttt{USDCHF}~\cite{islam2021foreign} receiving the most attention~\cite{ayitey2023forex}.
However, these studies overlook the relationships among currency triplets---discussed in \Cref{sec:intro}---despite their strong theoretical basis in economics~\cite{cui2020detecting,aiba2002triangular}.

Prior studies do not consider the macroeconomic factors of currencies comprehensively. 
Theories such as the IRP and purchasing power parity delineate how an FX rate for a currency pair relates to macroeconomic factors \cite{taylor1995economics}. 
To utilize these theories, prior studies incorporate macroeconomic features of a target FX pair, such as IRs \cite{amat2018fundamentals,pfahler2022exchange,cao2020deep,zhang2020predictability}, inflation rates \cite{pfahler2022exchange,zhang2020predictability}, and monetary policies \cite{pfahler2022exchange,cao2020deep}.
However, the macroeconomic parity valuation of a currency can vary depending on the counterpart currency, due to market frictions (e.g., creditworthiness \cite{tuckman2003interest}, banking regulation \cite{du2018deviations}). 
Thus, accurate FXRP requires comprehensively leveraging macroeconomic factors of both the target pair and other currencies---an aspect overlooked by prior learning-based studies.

The time lag between price observation and execution is inevitable in practice, but prior studies overlook it and its associated risks. 
Several studies show the existence of FX arbitrage opportunities in real markets~\cite{aiba2002triangular,fenn2009mirage,mahmoodzadeh2020spot}.  
To exploit them, prior approaches employ integer programming~\cite{soon2007currency}, the simulated bifurcation algorithm~\cite{tatsumura2020currency}, and a GNN~\cite{zhang2025efficient} to solve arbitrage optimization problems.  
However, these studies disregard the time lag between price observation and execution, rendering them susceptible to look-ahead bias and the risk of price movements between the observation and execution time points. 
To be practically applicable, such methods need to provide very accurate FXRP or address the risks stemming from the time lag.
In addition, \cite{zhang2025efficient} is prone to violating flow-in-flow-out constraints due to its use of a relaxed loss with a squared penalty to satisfy the constraints, and it uses synthetic data. 
Unlike the questionable assumptions of \cite{zhang2025efficient}, our problem formulation accounts for the time lag and FX rate stochasticity. 
Moreover, our method provably satisfies all constraints, utilizes predictions from our FXRP method to address the time lag and stochasticity, and is evaluated on real market data.

To address the limitations of the prior studies, we propose a two-stage GL approach.  
The first formulates FXRP as a GL problem and presents a GL method, capturing interrelationships among FX rates and comprehensive interest-rate–based valuations.  
Building on these predictions, the second formulates FXSA as a stochastic optimization problem that addresses the observation–execution time lag and its associated risks. It also proposes a GL method that provably satisfies the empirical arbitrage constraints.

\section{Notation and Preliminaries} 
\textbf{Notation.} 
We map all time points on trading dates to $\mathcal{T} = (0, \infty)$ while preserving their temporal order, such that each natural number corresponds to the end time point (24:00) of a trading date.  
We use this end time to represent the trading date itself. 
That is, for a time point $\tau \in \mathcal{T}$, the ceiling $\lceil \tau \rceil \in \mathbb{N}$ denotes the corresponding trading date. 
We let $t\in\mathbb{N}\subset\mathcal{T}$ and $\tau\in\mathcal{T}$ denote a trading date and a time point, respectively.

We consider a single agent investing in the FX market. 
Let $C$ be the set of investable currencies, where $o \in C$ denotes the agent’s home currency.  
Let $Y_{\tau, \Delta \tau, i}\in\mathbb{R}$ and $Y_{\tau, \Delta \tau, i}^\mathrm{f}\in\mathbb{R}$ denote the government bond and risk-free IRs of currency $i\in C$, respectively, quoted at time $\tau \in \mathcal{T}$ for the period from $\tau$ to $\tau + \Delta \tau$. 
For $\tau \in \mathcal{T}$ and $i,j \in C$, let $X_{\tau ij}>0$ be a random variable denoting the FX rate from $i$ to $j$ at $\tau$, defined as the amount of $j$ per unit of $i$. 
Let $E_\tau\subset C \times C$ be the set of currency pairs that can be directly exchanged at $\tau$. 
For each $(i,j)\in E_\tau$, let $w_{\tau ij} \geq 0$ denote the amount, in units of currency~$o$, that a trading strategy allocates to exchange from $i$ to $j$.

We use the tilde symbol to denote the realization of a random variable, and the hat symbol to denote the agent’s estimate of its expectation.  
For any vectors $\mathbf{x}_i \in \mathbb{R}^p$ with $i \in I$, where $I$ is a finite index set, we denote by $[\mathbf{x}_i]_{i \in I} \in \mathbb{R}^{|I| \times p}$ the matrix formed by stacking $\mathbf{x}_i$ row-wise.  
For any $a, b \in \mathbb{R}$, we define $[a]^+ = \max\{a, 0\}$. We let $[a:b] = [a, b] \cap \mathbb{Z}$ and use similar notation for open or half-open intervals.

\textbf{Preliminaries.} 
An FX direct arbitrage refers to a riskless cash flow arising from a violation of  
\begin{equation}
    X_{\tau ij} X_{\tau ji} = 1 \label{eq:direct_arb}
\end{equation}
for $\tau\in\mathcal{T}$ and $i,j\in C$. 
If $X_{\tau ij} X_{\tau ji} > 1$ is known in advance, the agent can earn a profit of $X_{\tau ij} X_{\tau ji} - 1$ by converting one unit of currency~$i$ into $j$ and immediately  converting it back to~$i$ at $\tau$. 
Throughout this paper, we assume \cref{eq:direct_arb} holds to focus on more intriguing arbitrages involving more than two currencies. It is easy to incorporate such two-currency arbitrages in what follows.

An FX triangular arbitrage refers to a riskless cash flow that arises from the violation of  
\begin{equation}
    X_{\tau ij}=X_{\tau ik} X_{\tau kj} \label{eq:triangular_arb}
\end{equation}
for $\tau \in \mathcal{T}$ and $i,j,k \in C$. 
If both \cref{eq:direct_arb} and $X_{\tau ij}<X_{\tau ik}X_{\tau kj}$ hold and the agent knows this in advance, then, since $X_{\tau ik}X_{\tau kj}X_{\tau ji} \allowbreak > 1$, the agent can earn a profit of $X_{\tau ik}X_{\tau kj}X_{\tau ji} \allowbreak -1$ by sequentially exchanging from $i$ to $k$, from $k$ to $j$, and from $j$ to $i$. 
There is evidence that \cref{eq:triangular_arb} is violated in practice \cite{aiba2002triangular}, although it should hold in theory \cite{cui2020detecting}, implying underlying multi-currency relationships. 

Extending FX triangular arbitrage to cases involving more than three currencies, we consider an LP problem \cref{eq:prelim_arb_obj}--\cref{eq:prelim_nonneg}---adapted from \cite{soon2007currency}.  
For $\tau'', \tau', \tau \in \mathcal{T}$ with $\tau'' < \tau' < \tau$,  
suppose the agent observes $[\widetilde{X}_{\tau'' ij}]_{(i,j) \in E_{\tau''}}$ at time~$\tau'$ and knows, in advance at $\tau'$, that $\widetilde{X}_{\tau ij}= \widetilde{X}_{\tau'' ij}$ for all $(i, j) \in E_\tau \cap E_{\tau''}$. 
Then, the agent can estimate the expectation of $X_{\tau ij}$ as $\widehat{X}_{\tau ij}=\widetilde{X}_{\tau'' ij}$, solve the LP problem \cref{eq:prelim_arb_obj}--\cref{eq:prelim_nonneg}, and execute the solution at $\tau$ to maximize the profit from FX arbitrages. 
In \cref{eq:prelim_arb_obj}, the agent maximizes the profit defined as the difference between total cash inflows into $o$ and cash outflows from it.  
In \cref{eq:prelim_arb_flow_conserv}, the agent does not retain any nonzero amount in currencies other than $o$ to avoid the risk associated with holding a currency after trading, which implies the net cash flow for each $i \in C \setminus \{o\}$ is zero. 
The first and second terms in \cref{eq:prelim_arb_flow_conserv} represent the total inflows to $i$ and the total outflows from $i$, respectively. 
Since $w_{\tau ij}$ is in units of $o$, the second term converts it to $\widehat{X}_{\tau oi}w_{\tau ij}$ in units of $i$.
In the first term, the agent exchanges $\widehat{X}_{\tau oj}w_{\tau ji}$, in units of $j$, for $\widehat{X}_{\tau oj}\widehat{X}_{\tau ji}w_{\tau ji}$, in units of $i$. 
Equation \Cref{eq:prelim_arb_weight_sum} imposes the unit trading constraint. 
If a different amount is traded, the solution can be scaled proportionally. 
\begin{align}
    &\max_{w_{\tau ij}:(i,j)\in  E_\tau \cap E_{\tau''}} \sum_{j:(j,o)\in E_\tau \cap E_{\tau''}}  w_{\tau jo} - \sum_{j:(o,j)\in E_\tau \cap E_{\tau''}} w_{\tau oj} \label{eq:prelim_arb_obj}\\
    &~\text{s.t.} \notag\\
    & \sum_{j:(j,i)\in E_\tau \cap E_{\tau''}} \widehat{X}_{\tau oj} \widehat{X}_{\tau ji} w_{\tau ji} - \sum_{j:(i,j)\in E_\tau \cap E_{\tau''}} \widehat{X}_{\tau oi}w_{\tau ij} =0,\forall i \in C\setminus\{o\}, \label{eq:prelim_arb_flow_conserv}\\
    & \sum_{(i,j)\in E_\tau \cap E_{\tau''}} w_{\tau ij} = 1, \label{eq:prelim_arb_weight_sum}\\
    & ~w_{\tau ij}\geq 0, \forall (i,j) \in E_\tau \cap E_{\tau''}. \label{eq:prelim_nonneg}
\end{align}

Since the assumption $\widetilde{X}_{\tau ij}= \widetilde{X}_{\tau'' ij}$ in \cref{eq:prelim_arb_obj}--\cref{eq:prelim_nonneg} does not hold in practice, we need to predict $X_{\tau ij}$.  
To that end, we utilize IRs as input features based on the following IRP theories. 
Under the arbitrage-free assumption and ideal conditions~\cite{du2018deviations}, the covered IRP implies that for $\tau \in \mathcal{T}$ and $i,j\in C$,
$(1+Y_{\tau,\Delta \tau, i}^\text{f})  F_{\tau,\Delta \tau,i,j}=(1+Y_{\tau,\Delta \tau,j}^\text{f})  X_{\tau ij}$, 
where $F_{\tau,\Delta \tau, i, j}$ is the forward FX rate at $\tau$ from $i$ to $j$ with maturity at $\tau + \Delta \tau$.
If $F_{\tau,\Delta \tau, i, j}$ is replaced by the conditional expectation of $X_{\tau ij}$, then the covered IRP identity becomes the uncovered IRP identity.  
Both versions of IRP imply FX rates are theoretically related to IRs.

\section{Proposed Method}
To determine daily FXSA trade lists, we propose a two-step GL approach: one for FXRP and the other one for FXSA. 
The first yields FXRP model $f_P(\cdot;\theta_t^P)$, and the second yields FXSA model $f_S(\cdot;\theta_t^S)$ for $t\in[t_1:\infty)$, where $t_1\in\mathbb{N}$ is the date when the agent starts trading.
Here, $\theta_t^P$ and $\theta_t^S$ are trainable parameters. 
For each $t \in [t_1:\infty)$, $f_P(\cdot; \theta^P_t)$ predicts $\widehat{X}_{tij}$ for $X_{tij}$ using available information up to time $t - 1$. 
Then, $f_S(\cdot; \theta^S_t)$ determines trading quantities $w_{tij}$ based on the same information and the predictions $\widehat{X}_{tij}$ from $f_P$ by time $t - t_{\text{exec}}$, where $t_{\text{exec}} \in (0,1)$ denotes the time required to execute FX trades at $t$.  
The agent simultaneously executes all trades $w_{tij}$ at $t$, holds the resulting currencies until $t+1$, and converts all holdings to the home currency $o$ at $t+1$.

The FXRP and FXSA models are trained on certain trading dates and used until the next training date.  
Given hyperparameters $t_1$ and $ n_{\text{fit}} \in \mathbb{N}$, we define a sequence $\{t_k\}_{k=0}^{n_\text{fit}+1} \subset \mathbb{N} \cup \{\infty\}$ such that $t_0 = 1$, $t_k < t_{k+1}$ for all $k \in [0:n_\text{fit}]$, and $t_{n_\text{fit}+1} = \infty$.  
Then, we define the testing periods $\mathcal{T}_k^{\text{te}} = [t_k : t_{k+1})$ for all $k \in [1:n_\text{fit}]$, forming a partition of $[t_1 : t_{n_\text{fit}+1})$. 
For each $k\in [1:n_\text{fit}]$, during the period $(t_{k}-1,t_k-t_\text{exec})$, we train $f_P(\cdot; \theta^P_k)$ and $f_S(\cdot; \theta^S_k)$ using the training and validation periods $(\mathcal{T}_k^{P,\text{tr}},\mathcal{T}_k^{P,\text{va}})$ and $(\mathcal{T}_k^{S,\text{tr}},\mathcal{T}_k^{S,\text{va}})$, respectively, where $\mathcal{T}_k^{P,\text{tr}}\cup\mathcal{T}_k^{P,\text{va}}=\mathcal{T}_k^{S,\text{tr}}\cup\mathcal{T}_k^{S,\text{va}}=[t_0:t_k)$ and $\mathcal{T}_k^{P,\text{tr}}\cap\mathcal{T}_k^{P,\text{va}}=\mathcal{T}_k^{S,\text{tr}}\cap\mathcal{T}_k^{S,\text{va}}=\emptyset$. 
We then fix $\theta_t^P = \theta_k^P$ and $\theta_t^S = \theta_k^S$ for $t \in \mathcal{T}_k^{\text{te}}$.

\subsection{Foreign Exchange Rate Prediction}
\label{subsec:method_pred}

We formulate FXRP as an edge-level regression problem on a discrete-time spatiotemporal graph.  
For each $t \in \mathcal{T}$, we define a graph $\mathcal{G}_t^P=(C,\allowbreak E_t,\allowbreak [Y_{t,1,i}]_{i\in C},\allowbreak [X_{tij}]_{(i,j)\in E_t})$,  
whose components represent the node set, edge set, node features, and edge features, respectively.  
We define $\mathcal{G}_{1:t}^P$ as a collection $\{\mathcal{G}^P_{t'}:t' \in [1:t]\}$. 
We define $f_P(\mathcal{G}_{1:t-1}^P;\theta^P_t)=h_{PO}(g_{P}(h_{PI}(\mathcal{G}_{1:t-1}^P);\theta^P_t))=[\widehat{X}_{tij}]_{(i,j)\in \mathcal{U}_t}$ as the prediction function. 
Here, $h_{PO}$ and $h_{PI}$ are output-scaling and feature-engineering functions, respectively, and $g_P$ is a GNN with parameters $\theta_t^P$.
We formulate the FXRP problem as training $f_P$ to predict $\widetilde{X}_{tij}$ through the following loss:
\begin{align}
    &\min_{\theta^P_{k}} \sum_{t\in \mathcal{T}_k^{P,\text{tr}}} d_{P}\left( 
    g_P(h_{PI}(\mathcal{G}_{1:t-1}^P);\theta^P_{k}),
    h^{-1}_{PO}\left([\widetilde{X}_{tij}]_{(i,j)\in \mathcal{U}_t}\right)
    \right) \label{eq:pred:loss} 
\end{align}
for each $k\in[1:n_\text{fit}]$, 
where $d_{P}$ is a metric, and $\mathcal{U}_t\subset E_t$ is the set of predicted edges. 
We select the best among $n_\text{hyper}$ hyperparameter configurations via grid search, based on performance on $\mathcal{T}_k^{P,\text{va}}$. 
In our formulation, GNNs can capture multi-currency and currency–IR relationships through the graphs incorporating both FX rates and IRs.

To extract temporal information, we define $h_{PI}(\mathcal{G}_{1:t}^P)=\mathcal{G}_{t}^{PF}=(C,\allowbreak L_{t} \cap L_{t-1},\allowbreak 
[\textbf{c}_{ti}]_{i\in C}, \allowbreak 
[\textbf{x}_{tij}]_{(i,j)\in {L_{t} \cap L_{t-1}}})$. 
Here, we define $L_t=\{(i,j)\in E_t:(j,i)\in E_t\}$. 
As \cref{eq:direct_arb} is assumed, we consider only reciprocal edges, i.e., $L_t$ instead of $E_t$. 
We define $\mathcal{T}(i,j,t,t')=\{t''\in[t-t'+1:t]: (i,j) \in L_{t''} \cap L_{t''-1}\}$ 
as the set of dates $t''$ in $[t - t' + 1 : t]$ such that both exchanges $(i,j)$ and $(j,i)$ are directly tradable on both $t''$ and $t''-1$. 
We define $\mathcal{T}^\text{R}\subset\mathbb{N}$ as a set of look-back windows, which are hyperparameters.  
We define $\textbf{x}_{tij}=[x_{tijt'}]_{t' \in \mathcal{T}^\text{R}}$ 
where $x_{tijt'} = \frac{1}{|\mathcal{T}(i,j,t,t')|}\sum_{t''\in\mathcal{T}(i,j,t,t')}  \log ({\widetilde{X}_{t''ij}}/{\widetilde{X}_{t''-1,i,j}})$. 
Here, $x_{tijt'}$ denotes the average temporal log difference, $\log ({\widetilde{X}_{t''ij}}/{\widetilde{X}_{t''-1,i,j}})$, computed over the past $t'$ dates (i.e., $t''\in [t-t'+1:t]$) during which both $(i,j)$ and $(j,i)$ are directly tradable.  
The vector $\textbf{x}_{tij}$ collects these averages across different look-back windows $t'\in\mathcal{T}^\text{R}$. 
We define $[\textbf{c}_{ti}]_{i\in C}=[\textbf{y}_{ti} ; \textbf{v}_{ti}]\in\mathbb{R}^{2|\mathcal{T}^\text{R}|}$ where $\textbf{y}_{ti}=[y_{tit'}]_{t' \in \mathcal{T}^\text{R}}$ and  
$\textbf{v}_{ti}=[v_{tit'}]_{t' \in \mathcal{T}^\text{R}}$. 
Here, we define $y_{tit'}=\frac{1}{t'} \sum_{t''=t-t'+1}^{t} \allowbreak \log \left((1+\widetilde{Y}_{t'',1,i})/(1+\widetilde{Y}_{t''-1,1,i})\right)$ and $v_{tit'}=\frac{1}{t'}\sum_{t''=t-t'+1}^{t}\allowbreak\log \allowbreak ({\widetilde{V}_{t'',i}}/ \allowbreak {\widetilde{V}_{t''-1,i}})$ 
as the average temporal log differences of the interest rate and the currency value, respectively, over the past $t'$ dates, where $\widetilde{V}_{ti}$ is discussed next. 
These averages across different look-back windows $t'\in\mathcal{T}^\text{R}$ are collected in the vectors $\textbf{y}_{ti}$ and $\textbf{v}_{ti}$.

We define $\widetilde{V}_{ti}$ as the least-squares solution to \cref{eq:v:diff,eq:v:avg}:
\begin{align}
    & \log{\widetilde{V}_{ti}}-\log{\widetilde{V}_{tj}}=\log{\widetilde{X}_{tij}}, \forall (i,j)\in L_t: i<j, \label{eq:v:diff}\\
    & \frac{1}{|C|} \sum_{i\in C} \log \widetilde{V}_{ti} = 0, \label{eq:v:avg}
\end{align}
which is MLE in the market. 
Let $V_{ti}$ denote the value of currency $i$ in units of an imaginary currency $o^*$. 
From \cref{eq:triangular_arb}, for each $(i,j)\in L_t$, the exchange rate is modeled as $X_{tij}=\exp(\alpha_{tij}) \cdot V_{ti}/V_{tj}$, or equivalently,
$\log X_{tij}=\log V_{ti} - \log V_{tj}+ \alpha_{tij}$. 
Since $\alpha_{tij} \neq 0$ implies arbitrage, and such opportunities are expected to be exploited \cite{varian1987arbitrage,langenohl2018sources}, we assume $\mathbb{E}\alpha_{tij} = 0$. 
For analytical tractability, we assume $\alpha_{tij} \sim \mathcal{N}(0, \sigma_\alpha^2)$. 
MLE maximizes the log-likelihood of $\alpha_{tij}$: $\max~[-\frac{|L_t|}{2} \log(2\pi\sigma_\alpha^2)-\frac{1}{2\sigma_\alpha^2}\sum_{(i,j)\in L_t}(\log \widetilde{V}_{ti} - \log \widetilde{V}_{tj}-\log \widetilde{X}_{tij})^2]$. 
Its maximizer is the least-squares solution to \cref{eq:v:diff}, where the unknowns are $\log \widetilde{V}_{ti}$ and $\log \widetilde{V}_{tj}$.   
As \cref{eq:direct_arb} is assumed, we restrict \cref{eq:v:diff} to cases where $i < j$. 
Since the coefficient matrix of \cref{eq:v:diff} is a complete incidence matrix, it has rank $|C|-1$ \cite{bapat2014graphs}, and thus we add \cref{eq:v:avg}. 
Then, \cref{eq:v:diff,eq:v:avg} have rank $|C|$ \footnote{Since the left-hand side of \cref{eq:v:diff} is zero when $\log V_{ti}=1/|C|$ for all $i\in C$, the coefficient vector of \cref{eq:v:avg} is orthogonal to the row space of \cref{eq:v:diff}, and thus linearly independent of its rows.}, implying a unique least-squares solution always exists.

We define $h^{-1}_{PO}([X'_{tij}]_{(i,j)\in \mathcal{U}_t})=[\log(X'_{tij}/\widetilde{X}_{t-1,ij})]_{(i,j)\in \mathcal{U}_t}$ and define $h_{PO}$ as the inverse of $h^{-1}_{PO}$. 
We set $\mathcal{U}_t=L_t \cap L_{t-1} \cap L_{t-2}$ since $h_{PI}(\mathcal{G}_{1:t-1}^P)$ uses temporal differences between $t-1$ and $t-2$.

We use a simple GNN $g_P$ that handles both node and edge features. 
Given a graph $\mathcal{G}^{0}=(N,E,[\textbf{n}^0_i]_{i\in N},[\textbf{e}^0_{ij}]_{(i,j)\in E})$,   
we define $g_{C,l}(\mathcal{G}^{l-1})=\mathcal{G}^{l}=(N,E,[\textbf{n}^l_i]_{i\in N},[\textbf{e}^l_{ij}]_{(i,j)\in E})$ for each $l \in [1:L]$, where
\begin{align}
    &\textbf{n}^l_i= \frac{1}{|N(i)|}\sum_{j\in N(i)}\text{SLP}_{N,l}([\textbf{n}^{l-1}_i; \textbf{e}^{l-1}_{ji} ;\textbf{n}^{l-1}_j]), \forall i\in N, \label{eq:conv_node} \\
    &\textbf{e}^l_{ij} = \text{SLP}_{E,l}([\textbf{n}^l_i ; \textbf{e}^{l-1}_{ij} ; \textbf{n}^l_j]),\forall (i,j)\in E, \label{eq:conv_edge}
\end{align}
$N(i)=\{j:(j,i) \in E\}$, and $\text{SLP}$ denotes a single-layer feedforward network. 
Let $g_{CN,l}(\mathcal{G}^{0})=[\textbf{n}^l_i]_{i\in N}$ and $g_{CE,l}(\mathcal{G}^{0})=[\textbf{e}^l_{ij}]_{(i,j)\in E}$. 
Then, we define $g_P$ with $L$ layers as $g_P(h_{PI}(\mathcal{G}^P_{1:t-1}))\allowbreak=g_P(\mathcal{G}^{PF}_{t-1})=SLP_{H}(g_{CE,L}(g_{CS}(\mathcal{G}_{t-1}^{PF})))$, where $\text{SLP}_H$ has no activation function, and $g_{CS}$ scales the node and edge features of $\mathcal{G}_{t-1}^{PF}$.

By \cref{eq:conv_node,eq:conv_edge}, $g_P(h_{PI}(\mathcal{G}^P_{1:t-1}))$ can capture complex multi-currency–IR relationships.  
In \cref{eq:conv_node}, each currency aggregates information about the IRs of its neighbors and the FX rates connected to it. 
Using such information from two currencies, the embedding of their FX rate is updated in \cref{eq:conv_edge}. 
Over multiple layers of \cref{eq:conv_node,eq:conv_edge}, $g_P$ extracts complex relationships between FX rates and IRs.

Since $f_S$ uses the prediction $[\widehat{X}_{tij}]_{(i,j)\in\mathcal{U}_t}$, we impose an additional condition on $\mathcal{T}_k^{P,\text{va}}$ such that $\bigcup_{k=0}^{n_\text{sy}} \mathcal{T}_k^{P,\text{va}}= [t_0:t_1)$, where $n_\text{sy}\in[1:n_\text{fit})$ is a hyperparameter.  
Then, for $t\in [t_0,t_1)$, we let $k^*(t)=\min \{k'\in [1:n_\text{fit}):t\in\mathcal{T}_{k'}^{P,\text{va}} \}$ and define $[\widehat{X}_{tij}]_{(i,j)\in\mathcal{U}_t}=f_P(\mathcal{G}_{1:t-1}^P;\theta^P_{k^*(t)})$. 
This ensures that for all $k \in [n_\text{sy}:n_\text{fit}]$, $f_S$ can be trained only on predictions made outside the training period of $f_P$.  
Without this condition, $f_S$ could be trained on $f_P$'s own training predictions for $t \in [t_0:t_1)$, i.e., $[t_0:t_1)\setminus\bigcup_{k=1}^{n_{\text{fit}}} \mathcal{T}_k^{P,\text{va}}\neq\emptyset$. 
We evaluate $f_P$ on $\mathcal{T}_k^\text{te}$ for all $k \in [1:n_\text{fit}]$, and $f_S$ for $k \in [n_\text{sy}:n_\text{fit}]$.

\subsection{Foreign Exchange Statistical Arbitrage}
\label{subsec:statarb}
\textbf{Problem Formulation.} We formulate FXSA as the following stochastic optimization problem: for each $t\in\mathbb{N}$, 
\begin{align}
    & \max_{w_{tij}: (i,j)\in \mathcal{U}'_t} \quad \frac{\mathbb{E}G_t}{\sqrt{Var(G_t)}} \quad \text{s.t.} \label{eq:max_util_obj} \\
    & ~G_t = \sum_{i\in C} \frac{1+Y_{t,1,i}^\text{f}}{1+Y_{t,1,o}^\text{f}}X_{t+1,i,o} H_{ti}, \label{eq:max_util_gain}\\
    & ~H_{ti} = \sum_{j:(j,i)\in \mathcal{U}'_t} X_{tji}\widehat{X}_{toj}w_{tji} - \sum_{j:(i,j)\in \mathcal{U}'_t} \widehat{X}_{toi} w_{tij},\forall i \in C, \label{eq:max_util_holding}\\
    & ~\mathbb{E}H_{ti}=0,\forall i \in C\setminus\{o\}, \label{eq:max_util_flow_conservation}\\ 
    & ~\sum_{(i,j)\in \mathcal{U}'_t} w_{tij} = 1, \label{eq:max_util_weight_sum}\\
    & ~w_{tij} w_{tji} =0, \forall (i,j)\in \mathcal{U}'_t, \label{eq:max_util_direct_arb} \\
    & ~w_{tij} \geq 0, \forall (i,j)\in \mathcal{U}'_t. \label{eq:max_util_nonneg}
\end{align}
Compared to \cite{soon2007currency,tatsumura2020currency,zhang2025efficient}, our problem includes the random variables $X_{tji}, X_{t+1,i,o}$ in \cref{eq:max_util_gain,eq:max_util_holding} to account for the time lag between observation at $t-1$ and execution at $t$. 

Deterministic decision variable $w_{tij}$ is the amount (in units of $o$) the agent converts from currency $i$ to $j$.  
We define $H_{ti} \in \mathbb{R}$, for $i \in C \setminus \{o\}$, as the agent’s holding in currency $i$ from $t$ to $t+1$; $H_{to} \in \mathbb{R}$ as the instantaneous profit or loss in $o$ at $t$; and $G_t \in \mathbb{R}$ as the time-$t$ present-value gain from trading, in units of $o$.
Since \cref{eq:max_util_gain,eq:max_util_holding} involve $\widehat{X}_{toi}$ and $X_{t+1,i,o}$, we restrict the link set to $\mathcal{U}'_t\subset\{(i,j)\in\mathcal{U}_t:\{(o,i),(o,j)\}\subset \mathcal{U}_t\cup\{(o,o)\}, \{(i,o),(j,o)\}\subset E_t\cup\{(o,o)\}\}$.  
As \cref{eq:direct_arb} is assumed, we remove the direct arbitrage trades---where both $w_{tij}$ and $w_{tji}$ are positive---by imposing \cref{eq:max_util_direct_arb}.

In \cref{eq:max_util_obj}, our problem maximizes the information ratio as defined in \cite{sharpe1998sharpe}, assuming a zero benchmark expected return.
Since FX rates can change between $t$ and $t+1$, i.e., $X_{tij} \neq X_{t+1,ij}$, we account for the associated risk by maximizing the risk-adjusted return measured by the information ratio. 
Because we pursue a strategy that holds no position between $t$ and $t+1$, i.e., $H_{ti} = 0, \forall i \in C \setminus \{o\}$, as required for deterministic arbitrage, we assume $\mathbb{E}[\sum_{i \in C \setminus \{o\}} H_{ti} X_{tio} Y^\text{f}_{t,1,o}] = 0$. 
That is, we assume that the expected return of the benchmark that invests $H_{ti}$ in the risk-free asset between $t$ and $t+1$ is zero. 
Since $H_{to}$ is an instantaneous profit or loss at $t$, it can be invested through other strategies between $t$ and $t+1$.

The agent trades at $t$, holds the resulting currencies $\{H_{ti}:i\in C \setminus \{o\}\}$ until $t+1$, converts all holdings back to $o$ at $t+1$, and then to their time-$t$ present values, as represented by \cref{eq:max_util_gain}. 
We assume that, for each $t \in \mathbb{N}$, the agent has unlimited borrowing and lending capacity at the rate $Y_{t,1,i}^\text{f}$ in any currency $i \in C\setminus \{o\}$. 
Thus, during the holding period $(t, t+1)$, interest accrues, and each holding position becomes $(1 + Y_{t,1,i}^\text{f}) H_{ti}$, in units of $i$, at $t+1$. 
After converting all holdings to the home currency $o$ and then to their time-$t$ values, the gain $G_t$ from the trading strategy is given by \cref{eq:max_util_gain}. 
For notational convenience, we define $X_{t,o,o}=\widetilde{X}_{t,o,o}=\widehat{X}_{t,o,o}=1$ for all $t\in\mathbb{N}$. 
Then, \cref{eq:max_util_gain,eq:max_util_holding} can also be applied to $H_{to}$.

In \cref{eq:max_util_holding,eq:max_util_flow_conservation}, for $i\in C\setminus \{o\}$, the holding position $H_{ti}$ is the net cash flow, and its expectation is constrained to be zero.  
In \cref{eq:max_util_holding}, the first and second terms denote total cash inflows and outflows, respectively. 
Since $w_{tij}$ is in units of $o$, it is converted to $i$ as $\widehat{X}_{toi} w_{tij}$, forming the second term. 
For the first term, the agent exchanges $\widehat{X}_{toj} w_{tji}$ in currency $j$ for $i$, resulting in $X_{tji}\widehat{X}_{toj} w_{tji}$. 
When converting from $j$ to $i$, the amount in $j$ should be decided by the agent before the exchange, while the resulting amount in $i$ is determined by the market at execution. 
Thus, $\widehat{X}_{toj}$ is an agent’s estimate, while $X_{tji}$ is a random variable. 
Due to the stochasticity of $X_{tji}$, we use \cref{eq:max_util_flow_conservation} as a surrogate constraint to reduce the risk of FX rate changes between $t$ and $t+1$, while we ideally pursue a strategy satisfying $H_{ti}=0,\forall i\in C\setminus\{o\}$.

\textbf{Graph Learning.}  
To solve \cref{eq:max_util_obj}--\cref{eq:max_util_nonneg}, we train $f_S$ to learn an optimal strategy from realized samples via an empirical formulation.  
We define a GL model $f_S(\mathcal{G}_{1:t-1}^P,[\widehat{X}_{tij}]_{(i,j)\in \mathcal{U}'_t};\theta_t^S)=h_{SO}(g_S(h_{SI}(\mathcal{G}_{1:t-1}^P,\allowbreak [\widehat{X}_{tij}]_{(i,j)\in \mathcal{U}'_t});\allowbreak \theta_t^S))=[w_{tij}]_{(i,j)\in\mathcal{U}'_t}$ 
and train it by minimizing the empirical objective corresponding to \cref{eq:max_util_obj}--\cref{eq:max_util_holding}. 
Here, $h_{SI}$ and $g_S$ denote a feature-mapping function and GNN, respectively. 
The post-processing function $h_{SO}$ ensures that the output satisfies the empirical constraints corresponding to 
\cref{eq:max_util_flow_conservation}--\cref{eq:max_util_nonneg}. 
To satisfy \cref{eq:direct_arb}, we redefine $\widehat{X}_{tij}$ as the geometric average of the predictions from $f_P$: $\widehat{X}_{tij}$ and $1/\widehat{X}_{tji}$.

To achieve \cref{eq:max_util_obj}--\cref{eq:max_util_holding}, we train multiple GL models and select the best models. 
For each $k\in[1:n_\text{fit}]$, we train $f_S$ with loss  
\begin{align}
    &\min_{\theta_{k}^S} \quad \mathbb{E}_\mathcal{B} [-\frac{\hat{\mu}^2_{G}}{\hat{\sigma}^2_G}\cdot\mathbb{I}_{\{\hat{\mu}_{G}>0\}}- \hat{\mu}_{G}\cdot\mathbb{I}_{\{\hat{\mu}_{G}\leq0\}}]  \quad \text{s.t.} \label{eq:glarb:obj}\\
    &~\hat{\mu}_{G}=\frac{1}{|\mathcal{B}|}\sum_{t\in\mathcal{B}} \widetilde{G}_t,\quad \hat{\sigma}^2_G=\frac{1}{|\mathcal{B}|-1}\sum_{t\in\mathcal{B}}(\widetilde{G}_t-\hat{\mu}_{G} )^2, \label{eq:glarb:muhat}\\ 
    & ~\widetilde{G}_t = \sum_{i\in C} \frac{1+\widetilde{Y}_{t,1,i}}{1+\widetilde{Y}_{t,1,o}}\widetilde{X}_{t+1,i,o} \widetilde{H}_{ti}, \label{eq:glarb:gain}\\ 
    & ~\widetilde{H}_{ti} = \sum_{j:(j,i)\in \mathcal{U}'_t} \widetilde{X}_{tji}\widehat{X}_{toj}w_{tji} - \sum_{j:(i,j)\in \mathcal{U}'_t} \widehat{X}_{toi} w_{tij},\forall i \in C, \label{eq:glarb:holding}
\end{align}
where $[w_{tij}]_{(i,j)\in\mathcal{U}'_t}=f_S(\mathcal{G}_{1:t-1}^P,[\widehat{X}_{tij}]_{(i,j)\in \mathcal{U}_t};\allowbreak \theta_{kl}^S)$, and $\mathcal{B}\subset\mathcal{T}_k^{S,\text{tr}}$ is a batch. 
For each $k\in[1:n_\text{fit}]$, the best hyperparameter set is selected from among $n_\text{hyper}$ configurations based on performance on $\mathcal{T}_k^{S,\text{va}}$.  
Loss \cref{eq:glarb:obj} maximizes a squared estimate of \cref{eq:max_util_obj} when $\hat{\mu}_G > 0$, where $\mathbb{I}$ denotes the indicator function. 
Since maximizing ${\hat{\mu}^2_{G}}/{\hat{\sigma}^2_G}$ when $\hat{\mu}_{G}\leq0$ can increase the loss, we instead maximize the estimated profit. 
For $f_S$ to learn from realized samples, during training and evaluation, we use $\widetilde{X}_{t+1,i,o},\widetilde{X}_{tij},\widetilde{Y}_{t,1,i},\widetilde{Y}_{t,1,o}$ in lieu of $X_{t+1,i,o},X_{tij},Y_{t,1,i}^\text{f},Y_{t,1,o}^\text{f}$ in \cref{eq:max_util_gain,eq:max_util_holding}, yielding \cref{eq:glarb:gain,eq:glarb:holding}, respectively. 
Note that $\widetilde{Y}_{t,1,i}, \widetilde{Y}_{t,1,o}, \widetilde{X}_{t+1,i,o}$, and $\widetilde{X}_{tji}$ are not used by $f_S$ during inference---unlike in prior studies \cite{soon2007currency,tatsumura2020currency,zhang2025efficient}.

To handle \cref{eq:max_util_flow_conservation}--\cref{eq:max_util_nonneg}, 
we transform \cref{eq:max_util_holding,eq:max_util_flow_conservation} into empirical constraints and define $h_{SO}$. 
Replacing $\mathbb{E}X_{tij}$ with the predictions $\widehat{X}_{tij}$ from $f_P$,  
we convert the expectation of \cref{eq:max_util_holding} and \cref{eq:max_util_flow_conservation} into 
\begin{align}
    &\widehat{H}_{ti} = \sum_{j:(j,i)\in \mathcal{U}'_t} \widehat{X}_{tji}\widehat{X}_{toj}w_{tji} - \sum_{j:(i,j)\in \mathcal{U}'_t} \widehat{X}_{toi} w_{tij},\forall i \in C \label{eq:Hhat}, \\
    &\widehat{H}_{ti} =0,\forall i \in C \setminus\{o\}, \label{eq:Hhat_flow_conserv} 
\end{align}
respectively. 
While \cref{eq:max_util_weight_sum,eq:max_util_nonneg,eq:Hhat_flow_conserv} can be enforced by projecting the outputs of $g_S$ onto their linear constraint space, \cref{eq:max_util_direct_arb} is nonlinear.
Thus, we introduce $u_{tij}$ subject to 
\begin{align}
    &\sum_{j:(i,j)\in \mathcal{U}'_t} u_{tij} =0, \forall i \in C\setminus \{o\}, \label{eq:u:flow_conserv}\\ 
    &\widehat{X}_{toi} u_{tij} +\widehat{X}_{toj} \widehat{X}_{tji}u_{tji}  =0, \forall (i,j)\in \mathcal{U}'_t: i<j, \label{eq:u:sym} \\
    &u_{tij}\in\mathbb{R}, \forall (i,j)\in \mathcal{U}'_t, \label{eq:u:real} \\
    &w_{tij}=\frac{[u_{tij}]^+}{\sum_{(i,j)\in \mathcal{U}'_t}[u_{tij}]^+},\forall (i,j) \in \mathcal{U}'_t, \label{eq:u:w_u}
\end{align}
if $\sum_{(i,j)\in\mathcal{U}'_t}[u_{tij}]^+>0$. 
Then, due to \cref{eq:u:sym,eq:u:w_u}, $w_{tij}$ satisfies \cref{eq:max_util_direct_arb}.
In addition, \cref{eq:u:w_u} enforces \cref{eq:max_util_weight_sum,eq:max_util_nonneg}, and \cref{eq:u:flow_conserv} enforces \cref{eq:Hhat_flow_conserv}. 
The details are provided in the proof of \Cref{prop:one}.
We let $\widehat{D}_t=\{[u_{tij}]_{(i,j)\in\mathcal{U}'_t}: \text{\cref{eq:u:flow_conserv}--\cref{eq:u:real}}\} \subset \mathbb{R}^{|\mathcal{U}'_t|}$ and let $\widehat{B}_t$ denote a matrix whose columns form a basis for $\widehat{D}_t$.   
We then define $h_{SO}([u'_{tij}]_{(i,j)\in\mathcal{U}'_t})=[w_{tij}]_{(i,j)\in\mathcal{U}'_t}$ as in \cref{eq:u:w_u}, where
\begin{align}
    &[u_{tij}]_{(i,j)\in\mathcal{U}'_t}=\mathrm{Proj_{\widehat{D}_t}} [u'_{tij}]_{(i,j)\in\mathcal{U}'_t}=\widehat{P}_t [u'_{tij}]_{(i,j)\in\mathcal{U}'_t}, \label{eq:u:proj}
\end{align}
if $\sum_{(i,j)\in\mathcal{U}'_t}[u_{tij}]^+>0$. 
Here, $\widehat{P}_t=\widehat{B}_t(\widehat{B}_t^T \widehat{B}_t)^{-1}\widehat{B}_t^T$ is the projection matrix onto $\widehat{D}_t$. 
Since the projection in \cref{eq:u:proj} implies $[u_{tij}]_{(i,j)\in\mathcal{U}'_t}$ satisfies \cref{eq:u:flow_conserv}--\cref{eq:u:real}, the proposition below shows $h_{SO}$ guarantees the empirical constraints \cref{eq:Hhat_flow_conserv} and \cref{eq:max_util_weight_sum}--\cref{eq:max_util_nonneg}, which correspond to  \cref{eq:max_util_flow_conservation}--\cref{eq:max_util_nonneg}. 
Unlike \cite{zhang2025efficient}, our method does not allow constraint violations.

\begin{proposition}
\label{prop:one}
Equality~$C_1 = C_2$ holds if $\widehat{X}_{tij}\widehat{X}_{tji}=1$ for all $(i,j)\in\mathcal{U}'_t$, where
\begin{align*}
C_1 =\{[w_{tij}]_{(i,j)\in \mathcal{U}'_t}: & \text{\cref{eq:max_util_weight_sum}--\cref{eq:max_util_nonneg} and \cref{eq:Hhat_flow_conserv}}\},\\
C_2 =\{[w_{tij}]_{(i,j)\in \mathcal{U}'_t}: &\text{\cref{eq:u:flow_conserv}--\cref{eq:u:w_u}}, \textstyle \sum_{(i,j)\in\mathcal{U}'_t}[u_{tij}]^+>0 \}.   
\end{align*}
\end{proposition}
\begin{proof} 
First, we show $C_1\subseteq C_2$. 
Let $[w_{tij}]_{(i,j)\in\mathcal{U}'_t}\in C_1$.  
We define 
\begin{equation}
    u_{tij}=w_{tij}\mathbb{I}_{\{w_{tij}>0\}}-\frac{\widehat{X}_{toj} \widehat{X}_{tji}}{\widehat{X}_{toi}} w_{tji} \mathbb{I}_{\{w_{tij} \leq 0\}}. \label{eq:prop1:def_u}
\end{equation}
Then, \cref{eq:u:real} holds. 
By \cref{eq:max_util_weight_sum,eq:max_util_direct_arb}, there exists $(i,j)$ such that $w_{tij}>0$ and $w_{tji}=0$, implying $u_{tij}>0$. 
Thus, $\sum_{(i,j)\in\mathcal{U}'}[u_{tij}]^+ \allowbreak > 0$. 
From \cref{eq:prop1:def_u} and \cref{eq:max_util_nonneg}, $[u_{tij}]^+=w_{tij}$, which yields \cref{eq:u:w_u} by \cref{eq:max_util_weight_sum}.

By \cref{eq:Hhat}, $\widehat{H}_{ti}=\sum_{j:(i,j)\in \mathcal{U}'_t} (\widehat{X}_{tji}\widehat{X}_{toj}w_{tji} - \widehat{X}_{toi} w_{tij})$ since $(i,j)\in\mathcal{U}'_t$ if and only if $(j,i)\in\mathcal{U}'_t$.  
Since $w_{tij}\geq0$, we have $\widehat{H}_{ti}=\sum_{j:(i,j)\in \mathcal{U}'_t} (\widehat{X}_{tji}\allowbreak\widehat{X}_{toj}w_{tji}\mathbb{I}_{\{w_{tji}>0\}}  - \widehat{X}_{toi} w_{tij}\mathbb{I}_{\{w_{tij}>0\}})$.  
Since $w_{tji}\cdot\allowbreak\mathbb{I}_{\{w_{tji}>0\}}\allowbreak=w_{tji}\mathbb{I}_{\{w_{tij} \leq 0\}}$, 
by \cref{eq:max_util_direct_arb} and \cref{eq:max_util_nonneg}, we obtain $\widehat{H}_{ti}=\sum_{j:(i,j)\in \mathcal{U}'_t} (\allowbreak\widehat{X}_{tji}\widehat{X}_{toj}w_{tji}\cdot\allowbreak\mathbb{I}_{\{w_{tij} \leq 0\}}  - \widehat{X}_{toi} w_{tij}\mathbb{I}_{\{w_{tij}>0\}})=\allowbreak -\widehat{X}_{toi} \cdot \allowbreak \sum_{j:(i,j)\in\mathcal{U}'_t}u_{tij}$, which, together with \cref{eq:Hhat_flow_conserv}, yields \cref{eq:u:flow_conserv}.

By \cref{eq:prop1:def_u}, since $\widehat{X}_{tij}\widehat{X}_{tji}=1$, we obtain $\widehat{X}_{toi} u_{tij} +\widehat{X}_{toj} \widehat{X}_{tji}u_{tji}=\widehat{X}_{toi}w_{tij}(\mathbb{I}_{\{w_{tij}>0\}}-\mathbb{I}_{\{w_{tji} \leq 0\}})+\widehat{X}_{toj} \widehat{X}_{tji} w_{tji}(\mathbb{I}_{\{w_{tji}>0\}} -  \mathbb{I}_{\{w_{tij} \leq 0\}} )$.  
By \cref{eq:max_util_direct_arb} and \cref{eq:max_util_nonneg}, $w_{tij}(\allowbreak\mathbb{I}_{\{w_{tij}>0\}}\allowbreak-\mathbb{I}_{\{w_{tji} \leq 0\}})=w_{tji}(\mathbb{I}_{\{w_{tji}>0\}} -  \mathbb{I}_{\{w_{tij} \leq 0\}})=0$. 
Thus, we have \cref{eq:u:sym}.  Hence, $C_1\subseteq C_2$. 

Second, we show $C_2\subseteq C_1$. 
Let $[w_{tij}]_{(i,j)\in\mathcal{U}'_t}\in C_2$. 
From the definition of $C_2$, there exists $[u_{tij}]_{(i,j)\in \mathcal{U}'_t}$ such that $u_{tij}$ defines $w_{tij}$, \cref{eq:u:flow_conserv,eq:u:sym,eq:u:real,eq:u:w_u} hold, and $\sum_{(i,j)\in\mathcal{U}'_t}[u_{tij}]^+>0$. 
For $i \in C\setminus \{o\}$, 
\begin{align}
    & \widehat{H}_{ti}  \sum_{(i,j)\in \mathcal{U}'_t}[u_{tij}]^+ \\
    & = {\sum_{j:(j,i)\in \mathcal{U}'_t} \widehat{X}_{tji}\widehat{X}_{toj}\left[\frac{-\widehat{X}_{toi}u_{tij}}{\widehat{X}_{tji}\widehat{X}_{toj}}\right]^+} - {\sum_{j:(i,j)\in \mathcal{U}'_t} \widehat{X}_{toi} [u_{tij}]^+} \label{eq:proof_zero_node12}\\
    & =-\widehat{X}_{toi}\sum_{j:(i,j)\in \mathcal{U}'_t} \left( [ u_{tij}]^+-[-u_{tij}]^+\right) \label{eq:proof_zero_node13}\\
    & =-\widehat{X}_{toi}\sum_{j:(i,j)\in \mathcal{U}'_t} u_{tij}. \label{eq:proof_zero_node14}
\end{align}
We obtain \cref{eq:proof_zero_node12} from \cref{eq:Hhat,eq:u:w_u,eq:u:sym}.  
The definition of $[\cdot]^+$ yields \cref{eq:proof_zero_node13,eq:proof_zero_node14}. 
We have $\widehat{H}_{ti}  \sum_{(i,j)\in \mathcal{U}'_t}[u_{tij}]^+=-\widehat{X}_{toi}\sum_{j:(i,j)\in \mathcal{U}'_t} u_{tij}\allowbreak=0$, by \cref{eq:u:flow_conserv}. 
Since $\sum_{(i,j)\in\mathcal{U}'_t}[u_{tij}]^+>0$, \cref{eq:Hhat_flow_conserv} holds. 

By \cref{eq:u:w_u}, \cref{eq:max_util_weight_sum} holds. 
By combining \cref{eq:u:sym}, i.e., $u_{tij}u_{tji}\leq0$, and \cref{eq:u:w_u}, we obtain \cref{eq:max_util_direct_arb}. 
From \cref{eq:u:w_u}, \cref{eq:max_util_nonneg} holds. 
Hence, $C_2\subseteq C_1$. 
\end{proof}

We define $f_{SI}(\mathcal{G}_{1:t-1}^P,[\widehat{X}_{t,ij}]_{(i,j)\in \mathcal{U}'_t})=\mathcal{G}^{P''}_{t-1}=(\mathcal{U}'_{t},\widehat{A}_t, [\hat{\boldsymbol{\alpha}}_{tij}\allowbreak]_{(i,j) \in \mathcal{U}'_{t}}, [\hat{\textbf{p}}_{tij}]_{(i,j)\in A_t})$. 
Here, $\hat{\boldsymbol{\alpha}}_{tij}=[\bar{\alpha}_{tijt'}]_{t' \in \mathcal{T}^\text{R}}\in\mathbb{R}^{|\mathcal{T}^\text{R}|}$ 
where $\bar{\alpha}_{tijt'} = \frac{1}{|\mathcal{T}'(i,j,t,t')|} \sum_{t\in\mathcal{T}'(i,j,t,t')}  \hat{\alpha}_{tij}$, 
$\hat{\alpha}_{tij}=\log \widehat{X}_{tij}-\log \widehat{V}_{ti} + \log \widehat{V}_{tj}$, and $[\log\widehat{V}_{ti}]_{i \in C}$ is the least-squares solution to \cref{eq:v:diff,eq:v:avg} when $\widetilde{V}_{ti},\widetilde{V}_{tj},\widetilde{X}_{tij}$ are replaced by $\widehat{V}_{ti},\widehat{V}_{tj},\widehat{X}_{tij}$, respectively.
We define $\mathcal{T'}(i,j,t,t')=\{t''\in[t-t'+1:t]: (i,j) \in L_{t''}\}$ as the set of dates $t''\in[t - t' + 1 : t]$ on which both exchanges $(i,j)$ and $(j,i)$ are directly tradable. 
As discussed in \cref{eq:v:diff,eq:v:avg},  $\hat{\alpha}_{tij}$ provides information about estimated arbitrages, and $\hat{\boldsymbol{\alpha}}_{tij}$ consists of temporal averages of $\hat{\alpha}_{tij}$ across different look-back windows in $\mathcal{T}^\text{R}$. 
We define $\widehat{A}_t=\{(i,j): |\hat{p}_{tij}|>\epsilon_S\}$, where $\hat{p}_{tij}$ is the $(i,j)$-th entry of $\widehat{P}_t$, and $\epsilon_S>0$ is a small number. 
We define $\hat{\textbf{p}}_{tij}=[\bar{p}_{tijt'}]_{t'\in\mathcal{T}^\text{R}}\in\mathbb{R}^{|\mathcal{T}^\text{R}|}$ where $\bar{p}_{tijt'} = \frac{1}{|\mathcal{T}'(i,j,t,t')|} \sum_{t\in\mathcal{T}'(i,j,t,t')} \hat{p}_{tij}$.
The edge features provide information about how much the target node's trading quantity changes as the source node's quantity changes. 
We define $g_S(\mathcal{G}^{P''}_{t-1})$ with $L$ layers as $\text{SLP}_{H}(g_{CN,L}(g_{CS}(\allowbreak\mathcal{G}^{P''}_{t-1})))$.

\section{Experiments}

We use Finaeon data \cite{finaeon_gfdatabase_2025} from Jan. 1, 1995 to Dec. 31, 2024, covering ten of the most traded currencies in 2022 \cite{BIS2022fx}: \texttt{USD}, \texttt{EUR}, \texttt{JPY}, \texttt{GBP}, \texttt{AUD}, \texttt{CAD}, \texttt{CHF}, \texttt{HKD}, \texttt{SGD}, and \texttt{SEK}. \texttt{CNY} is excluded due to incomplete data.
The data include daily closing exchange rates among the ten currencies and their IRs for 1-, 2-, 5-, and 10-year government bonds. 
To satisfy assumption \cref{eq:direct_arb}, FX rates $\widetilde{X}_{tij}$, excluding direct arbitrages, are redefined as the geometric average of $\widetilde{X}_{tij}$ and $1/\widetilde{X}_{tji}$. 
For temporal stability, missing values in $\widetilde{X}_{tij}$ are forward-filled for up to seven days only when computing \cref{eq:v:diff,eq:v:avg}. 
Missing IR values are forward-filled for up to 30 days; any remaining missing values are imputed via regression on $\log(\text{maturity})$.

To ensure data stability, we apply the following preprocessing.  
FX rates involving \texttt{JPY} and \texttt{SEK} that are mis-scaled by factors of 10, 100, or 10,000 are rescaled to their correct magnitudes. 
Erroneous values---those extremely large or small relative to adjacent time-series entries---are removed: \texttt{SGDAUD} has repeated entries from Dec.~29, 2014 to May~20, 2015; \texttt{SEKEUR} falls below 0.069 on Jun.~7, 2024; and \texttt{AUDCHF} exceeds 0.85 on Apr.~24, 2023, Jan.~12, 2024, and Jan.~26, 2024. 
These values are identified by comparison with historical charts from \texttt{Yahoo} \texttt{Finance}~\cite{yahoofinance_check}. 
IR entries above 90\% for \texttt{HKG} and \texttt{SGP} during 2022 and 2023 are removed, following~\cite{hkab_hibor,abs_k2_rates}.

Trading is assumed to occur only on weekdays, with weekends excluded from trading dates. 
To ensure at least a few thousand training graphs for GL models, we define $t_1$ and $t_k$ as the first trading days of 2015 and each subsequent quarter, respectively, and set $n_\text{fit} = 40$ and $n_\text{sy} = 5$. 
We conduct three FXSA experiments by setting each of the most-traded currencies \texttt{USD}, \texttt{EUR}, and \texttt{JPY} as home currency~$o$.   
We estimate $\widehat{Y}_{t,1,i}$ by dividing the 1-year government bond rate of $i$ by 365. 
We set $\mathcal{T}^\text{R}=\{1,3,5,10,15,20\}$. 
We set $n_\text{hyper}=9$ and consider hyperparameter configurations $(\#params,L) \in \{10^4,5\cdot10^4,10^5\}\times\{2,3,4\}$, where $\#\text{params}$ is the number of parameters in $g_P$ and $g_S$, and $L$ is the number of layers.  
Each SLP in $g_P$ and $g_S$ uses the same number of hidden neurons with LeakyReLU activation.  
Our implementation uses \texttt{Python}, \texttt{Gurobi}, \texttt{R}, along with \texttt{PyTorch}, \texttt{PyTorch-Geometric}, \texttt{PuLP}, \texttt{SciPy}, \texttt{NumPy}, \texttt{Pandas}, and \texttt{lawstat} libraries. 
Experiments are run on a workstation with i9-14900KF, 96GB RAM, and RTX 4070 Ti SUPER.

\textbf{Foreign Exchange Statistical Arbitrage.} 
In \Cref{tab:trading:metrics}, our FXSA method (GNN) achieves, on average, 61.89\% higher information ratio and 45.51\% higher Sortino ratio than the benchmark LP \cref{eq:prelim_arb_obj}--\cref{eq:prelim_nonneg}, consistently outperforming across all three home currencies. 
The metrics where our method outperforms on average are in bold.  
Both our method and the LP use the same predictions for all $\widehat{X}_{tij}$, obtained from the GNN with FX/IR reported in \Cref{tab:pred:avg_mse}, which we discuss shortly.
Our method maintains a consistently higher information ratio than the LP throughout most of the test period (\Cref{fig:rolling_ir}), indicating superior risk-adjusted performance.
Additionally, the sharp drops in 2020, 2022, and 2024 (\Cref{fig:rolling_ir}) likely account for geopolitical shocks unseen in the training data: the COVID-19 pandemic, the Ukraine–Russia war, and the war in the Middle East, respectively. 
Although the average annual return is slightly lower, GNN shows a steadier upward trend than LP (\Cref{fig:cum_pnl}).

The superior performance of our FXSA method stems from better risk management, in our opinion. 
By aiming to maximize the information ratio in \cref{eq:glarb:obj}, the method considers both return and risk. 
As a result, our approach achieves, on average, 52.23\% lower annual volatility and 44.77\% lower maximum drawdown (MDD) (\Cref{tab:trading:metrics}). 
The approach also maintains 30.09\% lower average holdings, $\sum_{i\in C} |\widehat{H}_{ti} \widehat{X}_{tio}|$. 
Our method consistently exhibits lower values than LP throughout most of the test period (\Cref{fig:rolling_holdings}), thereby reducing exposure to the risk of FX rate movements.
Although both methods enforce zero expected holdings (i.e., \cref{eq:prelim_arb_flow_conserv} for LP and \cref{eq:u:flow_conserv}--\cref{eq:u:real} with \Cref{prop:one} for our method), realized values may deviate; our method keeps them closer to zero by minimizing \cref{eq:glarb:obj}.
Lastly, the proposed method shows an 83.12\% lower average Herfindahl-Hirschman index (HHI), indicating more diversified trading and thus reduced risk.
These outperformance results are consistent across the three home currencies.

\begin{table}
\caption{Evaluation Metrics of Trading Strategies (in \%)}
\label{tab:trading:metrics}
\begin{tabular}{c|rrr|rrr|r}
\toprule
Model & \multicolumn{3}{c|}{LP} & \multicolumn{3}{c|}{GNN} & \multicolumn{1}{r}{Change} \\
Currency $o$ & USD & EUR & JPY & USD & EUR & JPY & \multicolumn{1}{c}{Avg.} \\
\midrule
Info. Ratio & 27.21 & 27.10 & 26.87 & 43.86 & 43.54 & 44.01 & \textbf{61.89} \\
Sort. Ratio & 32.70 & 32.46 & 32.53 & 47.44 & 47.67 & 47.03 & \textbf{45.51} \\
Return & 6.19 & 6.18 & 6.13 & 4.80 & 4.74 & 4.75 & -22.73 \\
Volatility & 1.41 & 1.41 & 1.41 & 0.68 & 0.67 & 0.67 & \textbf{-52.23} \\
MDD & 1.61 & 1.60 & 1.72 & 0.92 & 0.89 & 0.91 & \textbf{-44.77} \\
HHI & 31.60 & 31.57 & 31.66 & 5.38 & 5.32 & 5.31 & \textbf{-83.12} \\
Holding & 0.50 & 0.49 & 0.46 & 0.34 & 0.34 & 0.33 & \textbf{-30.09} \\
\bottomrule
\end{tabular}
\end{table}
\begin{figure}[t]
  \centering
  \includegraphics[width=0.99\linewidth, keepaspectratio]{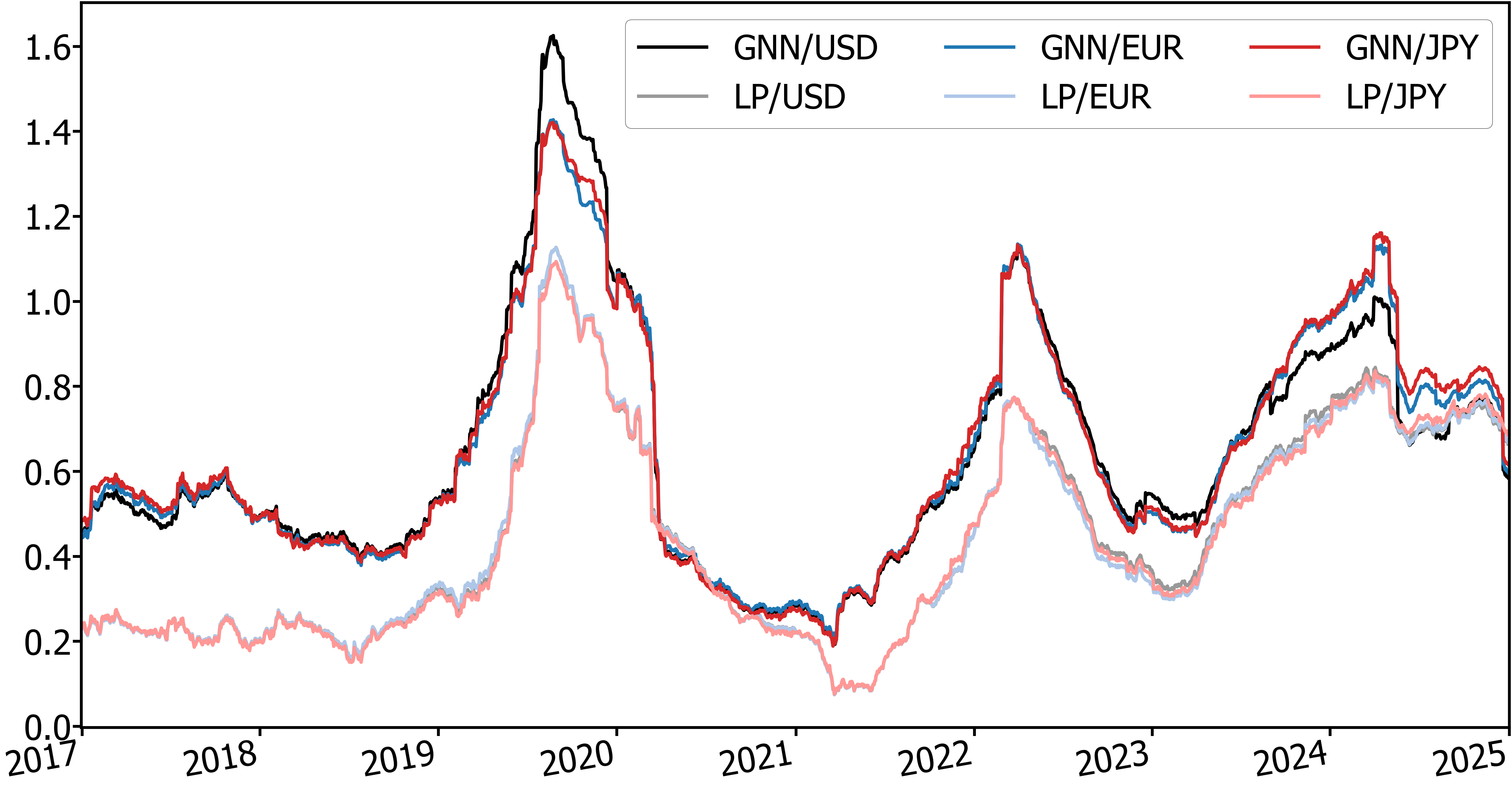}
  \caption{365-Day Rolling Information Ratio}
  \label{fig:rolling_ir}
\end{figure}
\begin{figure}[t]
  \centering
  \includegraphics[width=0.99\linewidth, keepaspectratio]{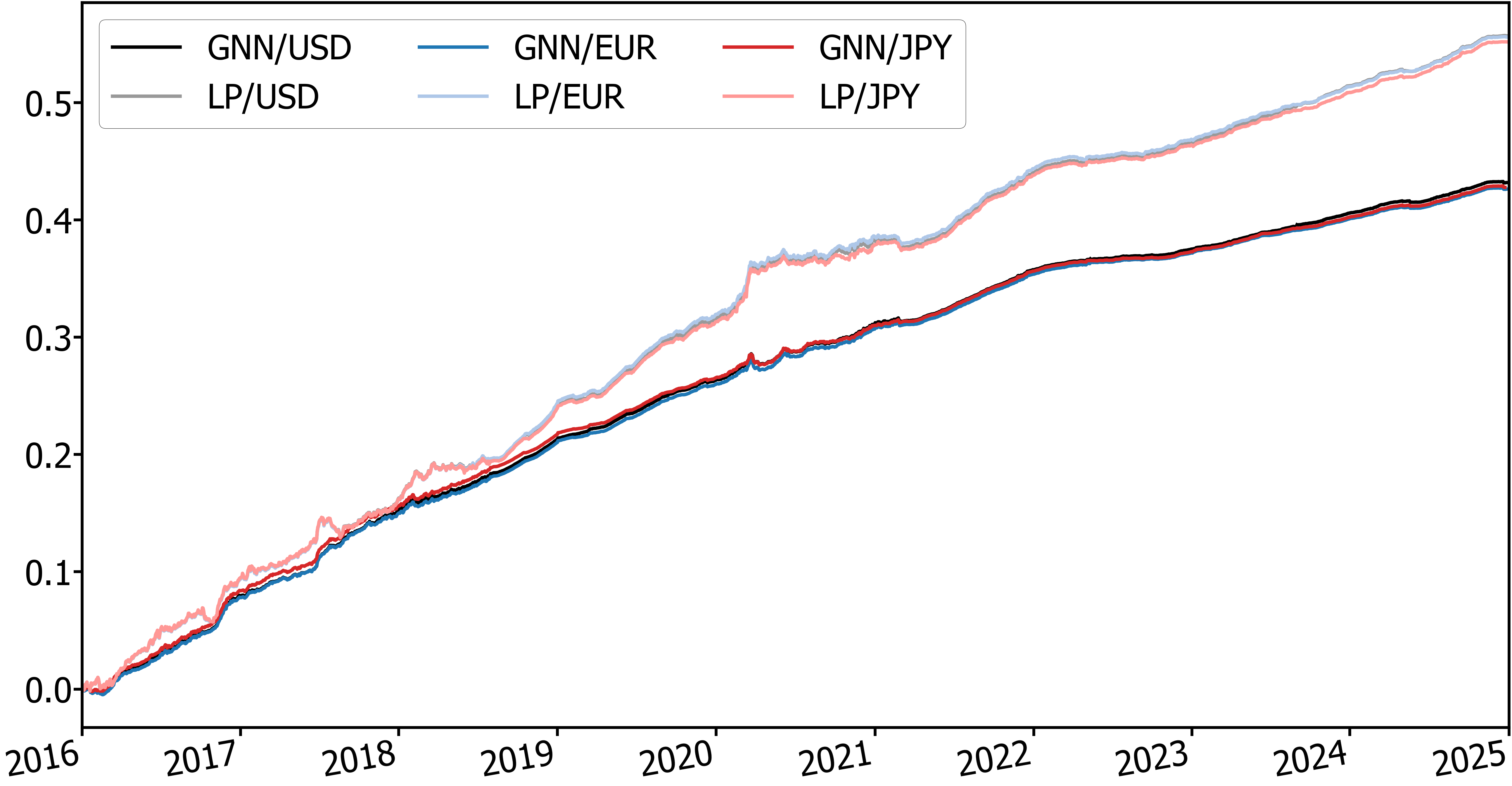}
  \caption{Cumulative Sum of P\&L ($G_t$), in Units of $o$}
  \label{fig:cum_pnl}
\end{figure}
\begin{figure}[t]
  \centering
  \includegraphics[width=0.95\linewidth, keepaspectratio]{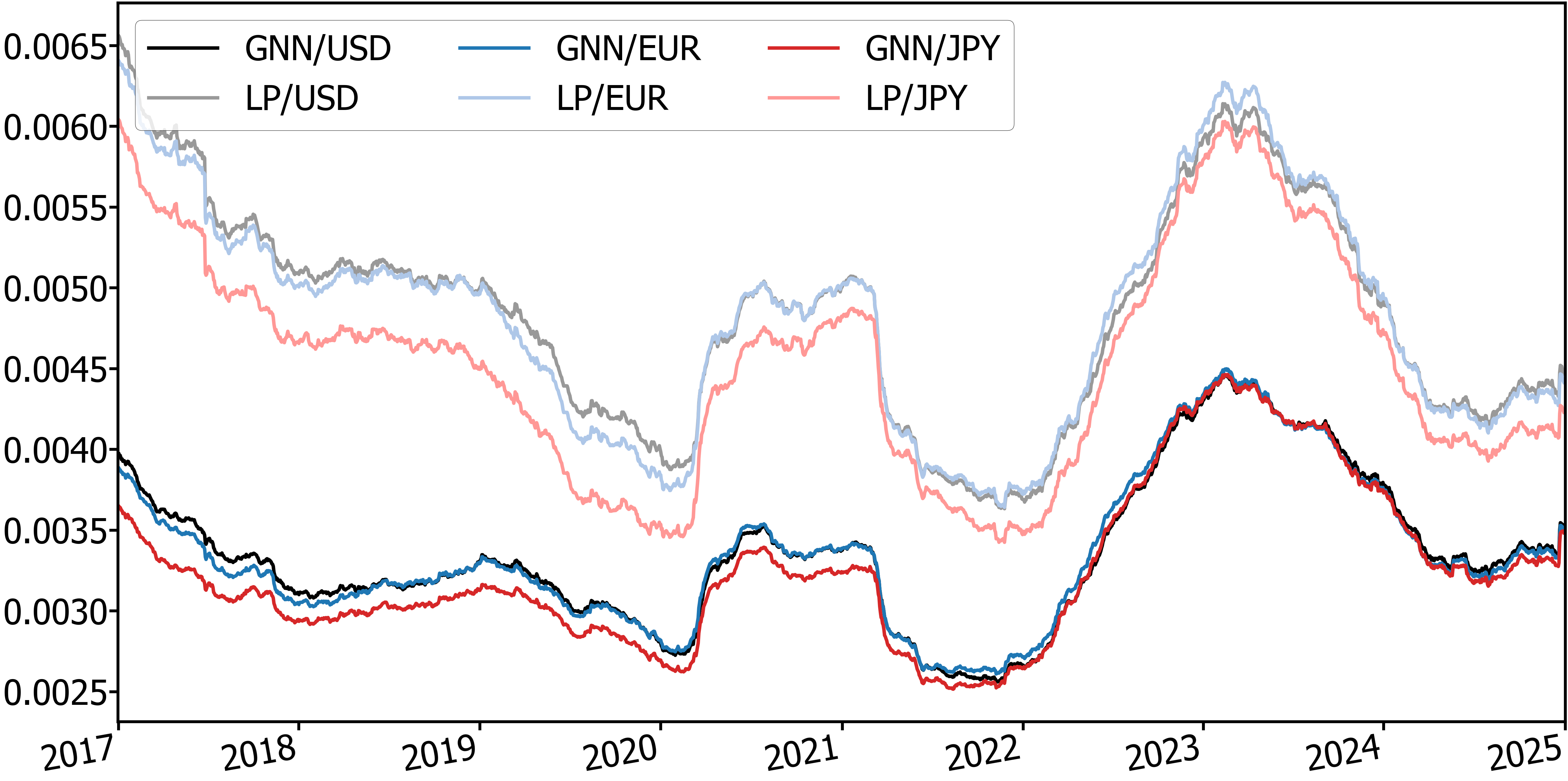}
  \caption{365-Day Rolling Avg. of $\sum_{i\in C}|\widetilde{H}_{ti}\widetilde{X}_{tio}|$, in Units of $o$}
  \label{fig:rolling_holdings}
\end{figure}

\textbf{Foreign Exchange Rate Prediction.} 
\Cref{tab:pred:avg_mse} shows our FXRP method outperforms the benchmark MLP in mean squared error (MSE). 
FX, CV and IR correspond to $\textbf{x}_{tij}$, $\textbf{v}_{ti}$, and $\textbf{y}_{ti}$ features in $\mathcal{G}^{PF}_t$, respectively.  
For each  $\mathcal{T}_k^\text{te}$ ($k\in[1,n_\text{fit}]$), we compute an MSE, and each value in the second and third rows reports the average of the MSEs across all $k$ for the corresponding model trained with the feature combination. 
All GNN results outperform their MLP counterparts, with the worst GNN outperforming the best MLP.

The outperformance of our GNN is statistically validated. 
In \Cref{tab:pred:avg_mse}, each value in the fourth to ninth rows is the p-value of a statistical test comparing the GNN and MLP under the same feature set. 
Wilcoxon, KS, and symmetry refer to the Wilcoxon signed-rank test, Kolmogorov–Smirnov test, and the symmetry test from \cite{miao2006new}, respectively. 
The paired t-, permutation, and Wilcoxon tests evaluate the one-sided hypothesis $H_0:\mathrm{MSE}(MLP) \leq \mathrm{MSE}(GNN)$ versus $H_1:\mathrm{MSE}(MLP) > \mathrm{MSE}(GNN)$, where the paired t-test and permutation test assess mean differences, and the Wilcoxon test assesses median differences.
To validate the assumptions behind these tests, we use the Shapiro–Wilk, KS, and symmetry tests. 
Bold p-values indicate statistically significant results (at the 0.05 level) where test assumptions hold. 
Thus, our GNN significantly outperforms the MLP across all feature sets. 
Moreover, \Cref{tab:pred:p_values} shows the GNNs statistically significantly outperform the MLPs across all feature combinations, except when using FX features alone.

\Cref{tab:pred:avg_mse} demonstrates the necessity of comprehensive relational information among FX rates and IRs captured by GNNs.
Since the CV features are derived from multiple FX differences in \cref{eq:v:diff,eq:v:avg}, they provide some relational information about FX rates, which may explain why the MLP with FX/IR/CV outperforms the MLP with FX/IR. 
However, CV captures limited relational information compared to the GNNs, as GNNs across all feature combinations outperform the MLP with FX/IR/CV.
Moreover, the GNN with FX/IR/CV underperforms the GNN with FX/IR, indicating that CV does not provide additional useful information beyond what the GNN with FX/IR captures.
Lastly, the best performance of the GNN with FX/IR underscores the importance of its information about complex multi-currency-IR relationships.

\begin{table}
\caption{MSE Comparison Between GNN and MLP Across Feature Combinations}
\label{tab:pred:avg_mse}
\begin{tabular}{llcccc}
\toprule
 &  & FX & FX/CV & FX/IR & FX/IR/CV \\
\midrule
\multirow[t]{3}{*}{\makecell[tl]{Avg. MSE\\(Unit: $10^{-5}$)}} & GNN & 2.8680 & 2.8648 & \textbf{2.8635} & 2.8650 \\
 & MLP & 2.9112 & 2.8704 & 2.9083 & 2.8710 \\
\cline{2-6}
 & Avg. & 2.8896 & 2.8676 & 2.8859 & 2.8680 \\
\cline{1-6}
\multirow[t]{6}{*}{P-Value} & Paired T & 0.0000 & 0.0082 & 0.0000 & \textbf{0.0054} \\
 & Wilcoxon & \textbf{0.0000} & \textbf{0.0112} & \textbf{0.0000} & 0.0045 \\
 & Permutation & 0.0000 & 0.0067 & 0.0000 & 0.0052 \\
\cline{2-6}
 & Shapiro & 0.0025 & 0.0004 & 0.0013 & 0.2863 \\
 & KS & 0.2390 & 0.1154 & 0.0709 & 0.7923 \\
 & Symmetry & 0.1440 & 0.0620 & 0.0880 & 0.5600 \\
\bottomrule
\end{tabular}
\end{table}
\begin{table}[t]
\caption{P-Values for MSE Comparison Between GNN and MLP Across Feature Combinations, Testing $H_0:\mathrm{MSE}(MLP) \leq \mathrm{MSE}(GNN)$ vs. $H_1:\mathrm{MSE}(MLP) > \mathrm{MSE}(GNN)$. \textnormal{Final letters T, W, and P indicate the paired t-, Wilcoxon signed-rank, and permutation tests, respectively, applied when their assumptions hold. 
}}
\label{tab:pred:p_values}
\begin{tabular}{rrrrrr}
\toprule
 &  & \multicolumn{4}{c}{GNN} \\
 &  & FX & FX/CV & FX/IR & FX/IR/CV \\
\midrule
\multirow[t]{4}{*}{MLP} & FX & 0.000 W & 0.000 W & 0.000 W & 0.000 W \\
 & FX/CV & 0.147 P & 0.011 W & 0.001 W & 0.003 W \\
 & FX/IR & 0.000 W & 0.000 W & 0.000 W & 0.000 P \\
 & FX/IR/CV & 0.113 W & 0.005 P & 0.000 W & 0.005 T \\
\bottomrule
\end{tabular}
\end{table}

\section{Conclusion}

This paper proposes a two-stage GL approach to generate FXSAs, addressing two key limitations in the literature: 
(i) the absence of GL methods for FXRP that leverage multi-currency and currency–IR relationships, and 
(ii) the disregard of the time lag between price observation and trade execution. 
In the first stage, we formulate FXRP as an edge-level regression problem on a discrete-time spatiotemporal graph to capture complex multi-currency and currency–IR relationships. 
We present a GL method to address the FXRP problem, leveraging the spatiotemporal graph defined with currencies and exchanges as nodes and edges, respectively, where IRs serve as node features, and FX rates serve as edge features. 
In the second stage, we present a stochastic optimization problem to exploit FXSAs while accounting for the observation-execution time lag. 
To solve this problem, we propose a GL method that enforces constraints through projection and ReLU to maximize the information ratio, while utilizing the predictions from the FXRP method for the constraint parameters and node features. 
We prove our GL method satisfies the empirical arbitrage constraints. 
Our GL method for FXSA uses a graph in which nodes represent exchanges and edges encode their influences.
Experimental results show our FXRP method achieves statistically significant improvements over non-graph methods. 
Moreover, our FXSA method achieves a 61.89\% higher information ratio and a 45.51\% higher Sortino ratio than the benchmark.
Our approach presents a novel perspective on FX prediction and StatArbs through the lens of GL.

\bibliographystyle{ACM-Reference-Format}
\bibliography{00.sections/08.references}
\end{document}